\documentclass[hidelinks,reprint,aps,showpacs,superscriptaddress,nofootinbib]{revtex4-1}
\usepackage{changepage}

\usepackage{graphicx}
\usepackage{amsmath} %
\usepackage{amssymb}  %
\usepackage{amsthm}
\usepackage{bm}
\usepackage{thmtools,thm-restate}

\usepackage{mathrsfs} %
\usepackage{stmaryrd}      
\usepackage{txfonts} %
\usepackage{lmodern}
\usepackage{mathtools, nccmath}
\usepackage{natbib}
\usepackage{enumerate}

\usepackage{appendix} 
\usepackage{xparse}
\usepackage[shortlabels]{enumitem}

\usepackage{chngcntr}
\usepackage{apptools}

\usepackage{environ}
\NewEnviron{killcontents}{}
\makeatletter{}%
\usepackage{graphicx}
\usepackage{amsthm}   %
\usepackage[bottom]{footmisc}%
\usepackage{amsmath} %
\usepackage{amssymb}  %
\usepackage{mathrsfs} %
\usepackage{stmaryrd} 
\usepackage{nicefrac}
  
\usepackage{enumerate}
\usepackage{appendix} 
\usepackage{enumitem}
\usepackage{comment}

\usepackage{mathtools}

\usepackage{xr}

\usepackage[capitalize]{cleveref}
\crefname{subsection}{Subsection}{Subsections}
\crefname{appendix}{Supplementary Note}{Supplementary Notes}

\usepackage{letltxmacro}
\LetLtxMacro{\originaleqref}{\eqref}
\renewcommand{\eqref}{Eq.~\originaleqref}

\usepackage{color}  %
\RequirePackage[dvipsnames,usenames]{xcolor}

\newif\ifnotes

\notesfalse

\def\shownotes{

}

\newcommand{\be}{\begin{equation}}
\newcommand{\bel}[1]{\begin{equation}\label{#1}}
\newcommand{\qe}{\end{equation}}
\newcommand{\ee}{\end{equation}}
\newcommand{\eeq}{\end{equation}}
\def\ba#1\ea{\begin{align}#1\end{align}}
\def\bann#1\eann{\begin{align*}#1\end{align*}}

\def\bal#1\eal{\begin{align}#1\end{align}}

\newcommand{\betight}{\begin{enumerate}[noitemsep,leftmargin=!,labelwidth=0em,labelindent=1em]
}
\newcommand{\eetight}{\end{enumerate}}

\newtheoremstyle{NatCommun}
{3pt}
{3pt}
{}
{}
{\bfseries}
{:}
{.5em}
{}
\theoremstyle{NatCommun}

\newtheorem{theorem}{Theorem}

\newtheorem{corollary}[theorem]{Corollary}

\newtheorem{example}{Example}

\newtheorem{lemma}[theorem]{Lemma}

\newtheorem{proposition}[theorem]{Proposition}

\newtheorem{definition}{Definition}

\shownotes

\newcommand{\cs}[1]{\mathsf{#1}}   

\DeclareMathOperator{\cycl}{cycl}
\DeclareMathOperator{\fix}{fix}
\DeclareMathOperator{\img}{img}

\DeclareMathOperator{\dom}{dom}

\newcommand{\map}{P}
\newcommand{\temp}{\mathrm{T}}

\newcommand{\auxstates}{k}
\newcommand{\auxsteps}{\ell}
\newcommand{\stablematrix}{L}
\newcommand{\adj}{\mathscr{A}}

\newcommand{\Z}{{\mathcal{Z}}}
\newcommand{\X}{{\mathcal{X}}}
\newcommand{\Y}{{\mathcal{Y}}}

\newcommand{\Cspace}[1]{C_\mathrm{space}(#1)}
\newcommand{\Ctime}[2][\auxstates]{C_{\mathrm{time}}(#2, #1)}

\NewDocumentCommand{\extrabit}{ O{\auxstates} O{f} }{b(#2,#1)}

\usepackage{blkarray}

\begin{document}

\title{A space-time tradeoff for implementing a function with master equation dynamics }

 \author{David H. Wolpert}
  \affiliation{Santa Fe Institute, 1399 Hyde Park Road,
  Santa Fe, NM 87501}
   \affiliation{Arizona State University}
 \affiliation{{\tt http://davidwolpert.weebly.com}}
   
 \author{Artemy Kolchinsky}
  \affiliation{Santa Fe Institute, 1399 Hyde Park Road,
  Santa Fe, NM 87501}

 \author{Jeremy A. Owen}
\affiliation{Physics of Living Systems Group, Department of Physics, Massachusetts Institute of Technology, 400 Tech Square, Cambridge, MA 02139.}


\begin{abstract}
Master equations are commonly used to model the dynamics of physical systems,
including systems that implement single-valued functions like a computer's update step.
However, many such functions
cannot be implemented by any master equation, even approximately, which raises the question of
how they can occur in the real world. Here
we show how any function over some ``visible'' states can be implemented with master equation dynamics --- if the dynamics exploits additional, ``hidden'' states at intermediate times.  
We also show that any master equation implementing a 
function can be decomposed into a sequence of ``hidden'' timesteps, demarcated by changes in what state-to-state transitions have nonzero probability. In many real-world situations there is a cost both for more hidden states and for more hidden timesteps. Accordingly, we derive a ``space-time'' tradeoff between the number of hidden states and the number of hidden timesteps needed to implement any given function.
\end{abstract}


\thanks{This is a post-peer-review version of an article published in \emph{Nature Communications}. The final, published version is available at \url{https://doi.org/10.1038/s41467-019-09542-x}}

\maketitle

\section*{Introduction}

Many problems in science and engineering involve understanding how
a physical system can implement a given map taking its initial, ``input'' state to its 
``output'' state at some later time. Often such a map is represented by some stochastic 
matrix $\map$. 
For example, $\map$ may be a conditional  distribution that governs 
the evolution of some naturally occurring system between two particular moments, and we wish to understand what
underlying physical process could result in that conditional distribution. Alternatively, $\map$ might represent some function $f$ that we wish to implement using a physical process,
e.g., $f$ could be the update function of the logical state of a digital computer.

In this paper we uncover constraints on the amounts of various resources 
that are needed by any system that implements a 
stochastic matrix $\map$. 
Throughout, we suppose that the underlying dynamics of the system  are continuous-time and Markovian.
(Such systems are sometimes said to evolve according to a ``master equation''.)
This basic assumption underlies many analyses in stochastic thermodynamics~\cite{seifert2012stochastic,barato2014unifying,esposito2012stochastic,horowitz_minimum_2017,sagawa2012fluctuation,riechers2017fluctuations}, and applies to many classical physical systems at the mesoscale, as well as semiclassical approximations of open quantum system with discrete states~\cite{van2013stochastic,esposito2010finite}.  Master equations also frequently appear in biology, demography, chemistry, computer science, and various other scientific fields. In addition
to assuming master equation dynamics, 
we focus on the case where $\map$ represents some single-valued function $f : \X \to \X$ over a finite space of ``visible states'' $\X$. For example, this would be the
case for any physical system that implements an
idealized digital device. 


The first resource we consider is the number of
``hidden states'' that are coupled to the states in $\X$ by the master
equation at intermediate times within the time interval $[0, 1]$. The second resource is the number of
successive subintervals of $[0, 1]$ which
are demarcated by moments when the set of state-to-state transitions allowed by the
master equation discontinuously changes. (We refer to each such subinterval as a ``hidden
timestep''.)

In the real world,
often it will be costly to have many hidden states and / or many hidden timesteps. 
For example, increasing the number of hidden states generally requires adding additional storage capacity to the system, e.g., by using additional degrees of freedom. 
Similarly, increasing the number of hidden timesteps carries a ``control cost'', i.e., it increases
the complexity of the control protocol that is used to drive the dynamics of the system. 
Moreover, 
transitions from one timestep to the next, during which the set of allowed state-to-state transition changes, typically require either the 
raising or dropping of infinite energy barriers between states in some underlying phase space~\cite{schnakenberg_network_1976,spohn_kinetic_1980,van_kampen_stochastic_1981,andrieux_fluctuation_2007,maroney2009generalizing}. 
Such operations typically require some minimal amount of time to be carried out.
Accordingly, the minimal number of hidden states and the minimal number of hidden timesteps
that are required to implement any given function $f$ can be viewed as fundamental ``costs of computation'' of a function $f$.

Physics has long been interested in the fundamental costs of performing computation and information processing. 
The most well-known of such costs is ``Landauer's bound'' \cite{landauer1961irreversibility,landauer1991information,bennett1982thermodynamics,sagawa2014thermodynamic,parrondo2015thermodynamics}, which states that 
the erasure of a physical bit, represented by 
a function $f : \{0,1\} \mapsto 0$,
requires the generation of at least $kT \ln 2$ heat when coupled to a heat bath at temperature $T$, assuming the initial value of the bit is uniformly distributed.
Recent studies have extended this bound to give the exact minimal amount of heat needed to
implement arbitrary functions $f$. 
These studies have all focused on implementing 
the given function $f$ with a physical system
whose dynamics can be approximated to arbitrary
accuracy with master equations~\cite{dillenschneider2009memory,sagawa2009minimal,berut2012experimental,diana2013finite,granger2013differential,jun2014high,zulkowski2014optimal,maroney2009generalizing,wolpert2015extending,wolpert_baez_entropy.2016}.
The two costs of computation proposed here arise, implicitly, in these previous analyses, since that the physical systems considered there 
all use hidden states.
However, none of these previous papers
considered the minimal number of hidden states needed to implement
a given function $f$ using master equations. (Rather  
they typically focused on issues related to thermodynamic reversibility.) 

In addition, the processes considered
in these papers all unfold through a sequence of distinct ``timesteps".  
In any single one of those timesteps, 
transitions between some pairs of states are allowed to occur while others are blocked, and the set of  allowed  transitions changes in going from one timestep to the next.
Again, despite their use of such hidden timesteps, none of these previous papers
considered the minimal number of hidden timesteps needed to implement
a function, given a certain number of available hidden states.

Our main results are exact expressions for the minimal number
of hidden states needed to implement a single-valued function $f$, and the
minimal number of hidden timesteps needed to implement $f$
given a certain number of hidden states. These results 
specify a tradeoff between the minimal number of hidden states and the minimal number of hidden timesteps required to implement a given $f$, which is analogous to the
``space-time'' tradeoffs that arise in the study of various models of computation in computer science.
However, here the tradeoff arises from the fundamental mathematical properties of continuous-time Markov processes. 
Moreover, real-world computers are constructed out of circuits, which are networks of computational elements  called  gates, each of which carries out a simple function.
For circuits, the tradeoff between hidden states and hidden timesteps 
that we uncover would apply in a ``local'' sense to the function carried out at each individual gate, 
whereas computer science has traditionally focused on ``global''
tradeoffs, concerning the set of all of those functions and of the network coupling them (e.g., the number of required gates or the ``depth'' of the circuit to compute some complicated $f$).

\section*{Results}

\subsection*{Markov chains and the embedding problem}
\label{sec:prelim}

We consider finite-state systems evolving under time-inhomogeneous continuous time Markov chains, 
which in physics are sometimes called ``master equations''. Such models of
the dynamics of systems are fundamental to many fields, e.g., they are very commonly used in stochastic thermodynamics~\cite{seifert2012stochastic,esposito2010three}. 
We begin in
this subsection by introducing some foundational concepts, which 
do not involve hidden states or hidden timesteps.  

We use calligraphic upper-case letters, such $\X$ and $\Y$, to indicate state spaces. 
We focus on systems with a finite state space.
We use the term {continuous-time Markov chain (CTMC)} $T(t,t')$ to refer to a set of transition matrices indexed by $ t\le t'$ which obey the Chapman-Kolmogorov equation $T(t,t') = T(t'',t')T(t,t'')$ for $t''\in [t,t']$. 
We use  {CTMC with finite rates} to refer to a CTMC such that the derivatives $\frac{d}{dt} T_{ij}(t,t')$ are well-defined and finite for all states $i,j$ and times $t\le t'$~\cite{doob_stochastic_1953}. 
For a given CTMC $T(t,t')$, we use $T_{ij}(t,t')$ to indicate the particular transition probability from state $j$ at time $t$ to state $i$ at time $t'$. 
Note that we do not assume time-homogeneous CTMCs, meaning that in general $T(t,t+\tau) \ne T(t', t'+\tau)$.  
Finally, note that the units of time are arbitrary in our framework, and for convenience we assume that $t=0$ at the beginning of the process and $t=1$ at the end of the process.   

The following definition is standard:
\begin{definition}
A stochastic matrix $\map$ is called {embeddable}
if $P = T(0,1)$ for some CTMC $T$ with finite rates.
\end{definition}
\noindent

As it turns out, many stochastic matrices cannot be implemented by any master equation.
(The general problem of finding a master equation that implements some given stochastic matrix $\map$ is known as the \emph{embedding problem} in the mathematics literature~\cite{kingman_imbedding_1962,lencastre2016empirical,jia2016solution}.)
One necessary (but not sufficient) condition
for a stochastic matrix $\map$ to be implementable with a master equation is~\cite{goodman_intrinsic_1970,kingman_imbedding_1962,lencastre2016empirical}
\ba
\prod_i \map_{ii} \ge \det \map > 0 \,.
\label{eq:det-condition}
\ea 
When $\map$ represents a single-valued function $f$ which is not the identity, $\prod_i P_{ii} = 0$ and the conditions of \cref{eq:det-condition} are not satisfied. 
Therefore, no non-trivial function can be exactly implemented with a master equation. 
However, as we show constructively in Supplementary Note 2, all non-invertible functions (e.g., bit erasure, which corresponds to $\map = \left( \begin{smallmatrix} 1&1\\ 0&0 \end{smallmatrix} \right)$) {can} be approximated arbitrarily closely using master equation dynamics. Intuitively, since the determinant of such functions equals $0$, they
can satisfy \cref{eq:det-condition} arbitrarily closely. 

To account for such cases, we introduce the following definition:

\begin{definition}
\label{def:limitembed}
A stochastic matrix $\map$ is {limit-embeddable} if there is 
a sequence of CTMCs with finite rates, $\{T^{(n)}(t, t') : n=1,2,\dots\}$, such that
\begin{align}
\map =\lim_{n\rightarrow\infty}T^{(n)}(0,1) \,.
\end{align}
\end{definition}
\noindent Note that while
each $T^{(n)}$ has finite rates, in
the limit these rates may go to infinity (this is sometimes called the ``quasistatic limit'' in physics). 
This is precisely what happens in the example of (perfect)  bit erasure, as shown explicitly in 
Supplementary Note 1.


We use the term {master equation} to broadly refer to a CTMC with finite rates, or the limit of a sequence of such CTMCs.

\subsection*{Definition of space and time costs}
\label{sec:defs}

When $\map$ represents a (non-identity) invertible function, 
$\prod_i P_{ii} = 0$, while 
$\det \map$ equals either $1$ or $-1$. So the conditions of \cref{eq:det-condition} are not
even infinitesimally close to being satisfied. This means that any (non-identity)
invertible function cannot be implemented, even approximately, with a master equation. As an example,  
the simple bit flip (which corresponds to the stochastic matrix $\map = \left( \begin{smallmatrix} 0&1\\ 1&0 \end{smallmatrix} \right)$), cannot be approximated by running 
any master equation over a two-state system.

How is it possible then that invertible functions can be accurately implemented by actual physical systems
that evolve according to a master equation?
In this paper we answer this question by showing that any function $f : \X \to \X$ over a set of {visible} states $\X$ can be implemented with a master equation --- as long as the master equation operates over a sufficiently large state space $\Y \supseteq \X$ that may include additional {hidden} states, $\Y \setminus \X$.
%
The key idea is that if $\Y$ is large enough,
then we can design the dynamics over the entire state  $\Y$ to be non-invertible, allowing the determinant condition of \cref{eq:det-condition} to be obeyed, while at the same time the desired
function $f$ is implemented over the subspace $\X$. 
As an illustration, below we explicitly show below how to implement a bit flip 
using a master equation over a 3-state system, i.e., a system 
with one additional hidden state. 

The following two definitions formalize what it means for one stochastic matrix to implement another stochastic matrix over a subset of its states. The first is standard.
\begin{definition}
The {restriction} of a $|\Y|\times|\Y|$ matrix $A$ to the set $\X \subseteq \Y$, indicated as $A_{[\X]}$, is the $|\X|\times|\X|$ submatrix of $A$ formed by only keeping the rows and columns of $A$ corresponding to 
the elements in $\X$.
\label{def:restriction}
\end{definition}

In all definitions below, we assume that $\map$ is a $|\X|\times|\X|$ stochastic matrix.

\begin{definition}
$M$ {implements $\map$ with $\auxstates$ hidden states} if $M$ is a $(|\X|+\auxstates)\times(|\X|+\auxstates)$ stochastic matrix and $M_{[\X]} = \map$.
\label{def:hidden_state}
\end{definition}

To see the motivation of \cref{def:hidden_state}, imagine that $M$ is a stochastic matrix implemented by some process, and $M_{[\X]} = \map$. If at $t=0$ the process is started in some state $i \in \X$, then the state distribution at the end of the process will be exactly the same as if we ran $\map$, i.e., $M_{ji} = \map_{ji}$ for all $j \in \X$.
Furthermore, because  $\sum_{j\in\X} \map_{ji} = 1$, $M_{ji} = 0$ for any $i \in X$ and $j \not\in \X$ (i.e., 
for any $j$ which is a  hidden state). This means that if the process is started in some $i\in \X$, no probability can ``leak" out into the hidden states by the end of the process, although it may pass through them at intermediate times.


The ``(hidden) space cost'' of $\map$ is
the minimal number of hidden states required to implement $\map$:

\begin{definition}
The {(hidden) space cost} of $\map$, 
written as $\Cspace{\map}$, is the smallest $\auxstates$ such that there exists a limit-embeddable matrix $M$ that implements $\map$ with $\auxstates$ hidden states.
\end{definition}

Consider a CTMC $T$ governing the evolution of a system. 
As $t$ increases, 
the set of transitions allowed by the CTMC (that is, the set of states which have $T_{ij}(0,t)>0$) changes. We 
wish to identify the number of such changes between $t=0$ and $t=1$
as the number of {``timesteps''} in $T$.  To formalize this, 
we first define the set of ``one-step'' matrices, which can be implemented by a CTMC which does not undergo 
any changes in the set of allowed transitions:
\begin{definition}
\label{def:singletimestep}
$\map$ is called {one-step} if $P$ is limit-embeddable with a sequence of CTMCs 
$\{T^{(n)} : n=1,2,\dots \}$ such that:
\begin{enumerate}
\item  $T(t,t'):=\lim_{n\rightarrow\infty}T^{(n)}(t,t')$ exists for all $t,t'\in[0,1]$;
\item $T(0, t)$ is continuous in $t \in (0, 1]$ and $T(t', 1)$ is continuous in $t' \in [0, 1)$; 
\label{enu:continuity}
\item For all $i,j$, either $T_{ij}(0,t) > 0$ for all $t\in(0,1)$, or $T_{ij}(0,t) = 0$ for all $t\in(0,1)$.
\label{cond:constantadjacnecy}
\end{enumerate}
\end{definition}

We note two things about our definition of one-step matrices. 
First, the precise semi-open interval used in the continuity condition (condition \ref{enu:continuity})  allows discontinuities in $T$ (and therefore in the set of allowed transitions) at the borders of the time interval. 
Second, we note that the limiting transition matrix $T$ in the above definition is still a CTMC. This is because: (1) a limit of a sequence of stochastic matrices is itself a stochastic matrix, so 
by definition $T(t,t')$ is a stochastic matrix for all $t, t' \in [0, 1]$, and (2) the Chapman-Kolmogorov equation $T(t,t') = T(t'',t')T(t,t'')$ holds (since it holds for each $T^{(n)}$). 
%
A canonical example of a one-step map is bit erasure,
as demonstrated in Supplementary Note 1. 

The definition of one-step matrices allows us to formalize the minimal number of  timesteps it takes to implement any given  $\map$:
\begin{definition}
The 
{(hidden) time cost with $\auxstates$ hidden states} of $\map$,
written as $\Ctime{P}$,
is the minimal number of one-step matrices of dimension $({|\X|+\auxstates}) \times ({|\X|+\auxstates})$
whose product implements $\map$ with $\auxstates$ hidden states.
\end{definition}
%
%
\noindent 
Note that a product of one-step matrices can be implemented
with a CTMC that successively carries out the CTMCs corresponding to each one-step matrix, 
one after the other.  So any stochastic matrix $\map$ with finite time cost can be
implemented as a single CTMC. Moreover, we can rescale units of time so that that product
of one-step matrices is implemented in the unit interval, $t \in [0, 1]$.
Note as well that since one-step matrices can have discontinuities at their borders,
the adjacency matrix of such a product of one-step matrices can
change from one such matrix to the next. 

\subsection*{The space-time tradeoff}
\label{sec:tradeoff}

For the rest of this paper we assume that our stochastic matrix of interest $\map$ is 0/1-valued, meaning that it represents a (single-valued) function $f : \X \to \X$. 
Below, in a slight abuse of previous notation, we will use $\Cspace{f}$ and $\Ctime{f}$ to refer to the space and time cost of implementing $f$. 
Except where otherwise indicated, all proofs are in the Methods section.

As we will show, there is a fundamental tradeoff between the number of available hidden states and the minimal number of timesteps. 
It will be convenient to present it using some standard terminology \cite{higgins1992}. 
For any function $f: X \to X$, we write  {$\fix (f)$} for the number of fixed points of $f$, and {$|\img (f)|$} for the size of the image of $f$. We also write {$\cycl (f)$} for the number of {cyclic orbits} of $f$, i.e., the number of distinct subsets of $X$ of the form
$\{x, f(x), f(f(x)), \ldots, x\}$ where $x$ is not a fixed point of $f$
and each element in the subset has a unique inverse under $f$.

We can now state our main result:

\begin{theorem}
For any single-valued function $f$ and number of hidden states $k$,
\begin{multline}
\Ctime{f}  = \\
\left\lceil\frac{\auxstates + |\X| + \max \! \big[\!\cycl (f) - \auxstates, 0\big] - \fix (f) }{\auxstates + |\X| - |\img (f) |}\right\rceil + \extrabit
\label{eq:mainresult}
\end{multline}
where $\lceil \cdot \rceil$ is the ceiling function and $\extrabit$ equals either zero or one (the precise value of $\extrabit$ is unknown for some functions).
\end{theorem}

Several corollaries from this result follow immediately:
\begin{corollary}
For any single-valued function $f$ and number of hidden states $k$,
\begin{equation}\label{c21} \Ctime{f} \approx \frac{|\X| + \cycl(f) - \fix(f)}{\auxstates + |\X| - |\img(f)|} + 1\end{equation}
and
\begin{equation} \Ctime{f} \le \frac{1.5 \times |\X|}{k} + 3 \,. \end{equation}



\label{corr:main_result}
\end{corollary}
\noindent
In addition, a ``converse'' of our main result gives $k_{\mathrm{min}}(f, \tau)$, the minimal number of
hidden states $\auxstates$ needed to implement $f$, assuming we are allowed to use at most $\tau$ timesteps.  
The exact equation for $k_{\mathrm{min}}(f, \tau)$ is presented in the Methods section. A simple 
approximation of that exact converse follows from \cref{corr:main_result}:
\begin{align}
k_{\mathrm{min}}(f, \tau) \approx \frac{\cycl(f) + |\X|(2-\tau) - \fix(f)}{\tau - 1} + |\img(f)| \,.
\end{align}

\begin{figure}
	\centering
	\includegraphics[width=1\columnwidth]{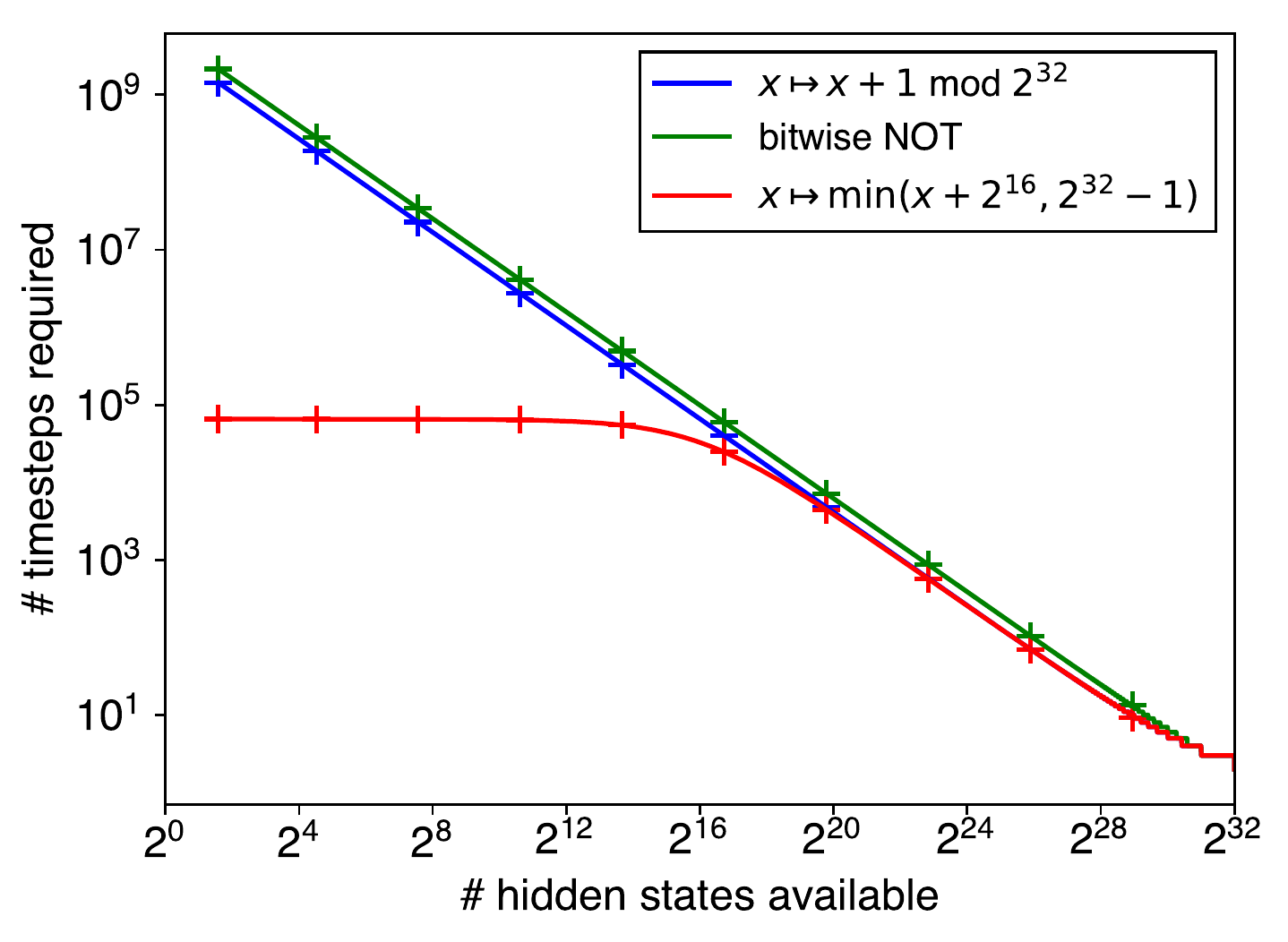}
	\caption{\textbf{The space-time tradeoff for three functions}. The domain of all three functions is $\X = \{0, \ldots, 2^{32}-1\}$. Solid lines show exact results, crosses indicate the approximation given by \eqref{c21}.
	}
	\label{fig:tradeoff}
\end{figure}

Although formulated in terms of time cost, our results have
some implications for space cost:
\begin{corollary}
For any non-invertible function $f$, $\Cspace{f} = 0$.
For any invertible $f$ (except the identity), $\Cspace{f} = 1$.
\label{coroll:simple}
\end{corollary}
\begin{proof}
If $f$ is non-invertible  $f$, $|\X| - |\img f| \neq 0$, so $\Ctime[0]{f}$ is finite.
Therefore, by definition of $C_\mathrm{space}$ and $C_\mathrm{time}$, 
$\Cspace{f} = 0$.
For invertible $f$, the denominator of \eqref{eq:mainresult} is zero 
if $k=0$. So while it is possible to implement any such $f$ (except the identity)
in a finite number of timesteps if we can use at least one hidden state, it
is impossible if we do not have any hidden states, i.e., 
$\Cspace{f} = 1$.
\end{proof}

\cref{fig:tradeoff} illustrates the tradeoff between space cost and time cost for 
three different functions over $\X = \{0, \dots, 2^{32}-1\}$. The first function (in blue) is an invertible ``cycle'' over the state space, computed as  $x \mapsto x +1 \; \mathrm{ mod }  \; 2^{32}$.  The second function (in green) is an invertible bitwise NOT operation, in which each element of $\X$ is treated as a 32-bit string and the value of each bit is negated. 
The third function (in red) is an addition followed by clipping to the maximum value, computed as $x \mapsto \min(x + 2^{16} , 2^{32}-1)$. Exact results (solid lines), as well as the approximation of  \eqref{c21} from \cref{corr:main_result} (crosses), are shown. 
These results show that achieving the minimal space costs given in \cref{coroll:simple}
may result in a very large time cost.

There are two important special cases of our result, which are analyzed in more detail in the Methods section.
First, when at least $\vert \img(f) \vert$ hidden states are available, any $f$ can be implemented in exactly two timesteps. 
Second, when $f$ is a cyclic permutation
and there is one hidden state available, the time cost is exactly $|\X| + 1$.

We emphasize that the proofs of these results (presented in the Methods section)
are constructive; for any choice of function $f$ and number of hidden states $k$, 
this construction gives a sequence of CTMCs with finite rates whose limit implements $f$
while using $k$ hidden states and the minimal number of hidden timesteps for that number of
hidden states. 
These constructions involve explicitly time-inhomogeneous master equations.
Indeed, for any time-homogeneous master equation, the set of allowed state transitions can never change, i.e., the only functions $f$ that can be implemented with such a master equation
are those that can be implemented in a single timestep.
Therefore our demonstrations of functions $f$ with time cost of $2$ proves
that there are maps that cannot be implemented unless one
uses a time-inhomogeneous master equation, no matter how many
hidden states are available.

\subsection*{Explicit constructions saturating the tradeoff}
\label{sec:construct}

We now illustrate our results using two examples. 
These examples use the fact that any idempotent function is one-step, as proved in \cref{lemma:mto_implementable} in the Methods section. (We remind the reader that a function $f$ is called idempotent if $f(x)=f(f(x))$ for all $x$.)

\begin{example}
\label{ex:capacitor}

Suppose we wish to 
implement the bit flip function $f:x \mapsto \neg x$ over $\X = \{0,1\}$. 
By \cref{coroll:simple}, since this map is invertible, 
we need exactly one hidden state to implement it. 

We introduce a space of three states $\Y = \{0, 1, 2\}$, and seek a sequence of idempotent functions
over $\Y$ that collectively interchange $0 \leftrightarrow 1$.
It is straightforward to confirm that our goal is met by the following sequence of idempotent functions: 
\begin{enumerate}
\item $\{1, 2\} \mapsto 2, \qquad 0 \mapsto 0$;
\item $\{0, 1\} \mapsto 1, \qquad 2 \mapsto 2$;
\item $\{0, 2\} \mapsto 0, \qquad 1 \mapsto 1$;
\end{enumerate}
Each idempotent can be implemented with the one-step CTMC described in Supplementary Note 2. 
This explicitly shows how to
implement a bit flip using one hidden state and three hidden timesteps.

Evaluating \cref{eq:mainresult} with $\auxstates =1$, $|\X|=|\img(f)| = 2$, $\cycl(f) = 1$, and $\fix(f)=0$ gives
\begin{equation}
\Ctime[1]{f} = 3 + \extrabit[1] \,.
\end{equation}
Thus, the above construction has optimal time cost (and, in this case, $\extrabit[1] = 0$).
\end{example}

\begin{figure*}
  \centering
    \includegraphics[width=0.65\textwidth]{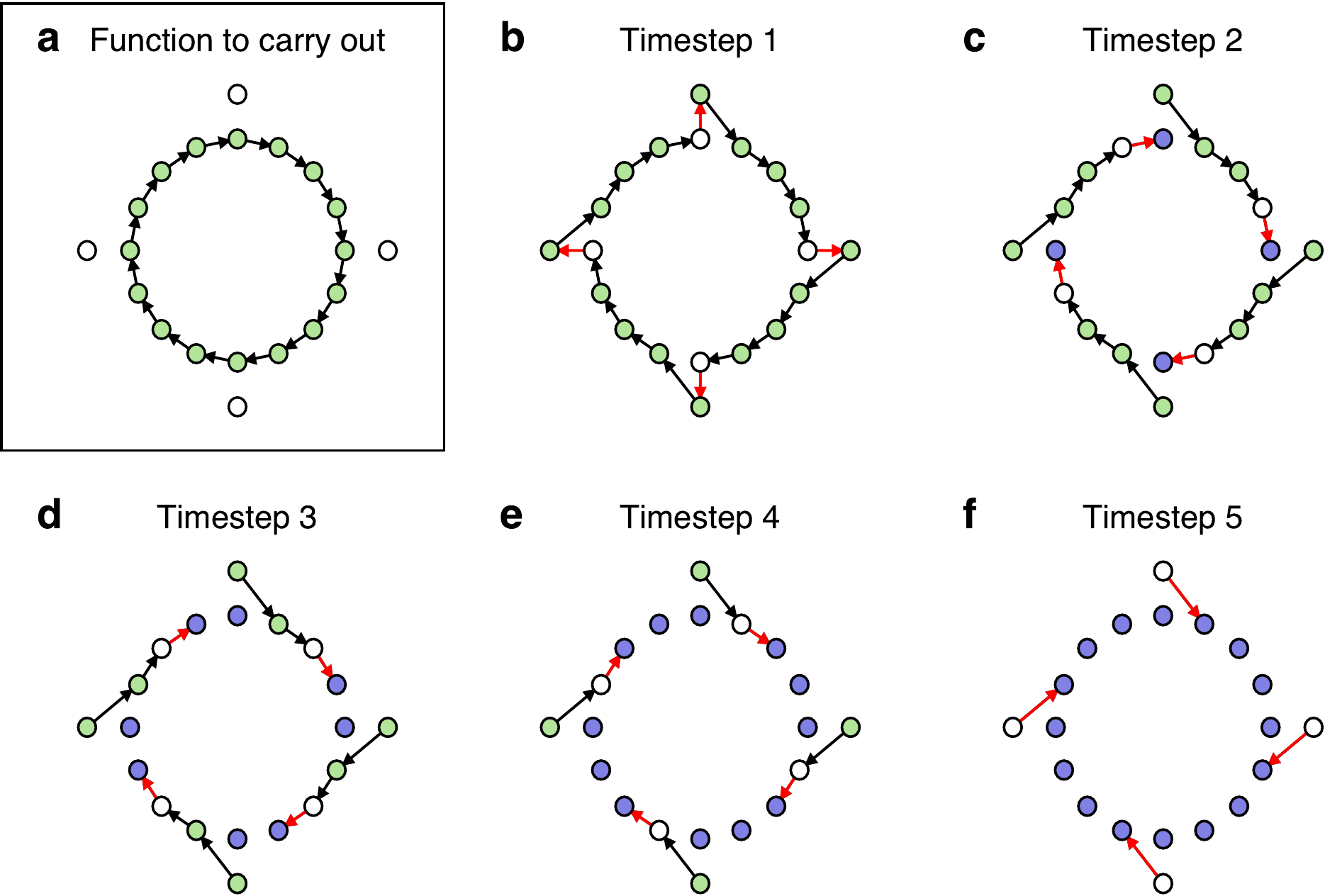}
  \caption{\textbf{Minimal-timestep implementation of a cyclic permutation with 4 hidden states}. The implementation carries out the function $f : x \mapsto x + 1 \text{ mod } 16$ over 16 states (green circles in \textbf{a}), using 4 hidden states (white
circles in \textbf{a}).   
  In all subplots, white nodes indicate states that cannot have any probability, light green  nodes with outgoing black arrows indicate states that may have positive probability but are not yet mapped to their final output, and purple nodes indicate states that may have positive probability and have been mapped to their final outputs. Subplots \textbf{b}--\textbf{f} show the state of the system after each of the 5 timesteps required to carry out $f$, where red arrows indicate the idempotent function carried out in each timestep (in each timestep, any state without outgoing red arrows is mapped to itself).
  \label{fig:cyclic-steps}
   }
\end{figure*}

The following example
demonstrates the implementation of a more complicated function, involving
a greater number of hidden states.

\begin{example}
\label{ex:cycle16}
Suppose we wish to implement the function
\begin{equation}
f(x) := x + 1 \; \mathrm{mod} \; 16
\end{equation}
over $\X = \{ 1 , \dots, 16 \}$.  For example, this kind of ``cyclic'' function may be used to keep track of a clock in a digital computer. Suppose also that 4 hidden states are available, so $\Y = \{1,\dots, 20\}$. The overall function to carry out, along with the hidden states, are shown in \cref{fig:cyclic-steps}A, along with a sequence of 5 idempotent functions over $\Y$ that carries out $f(x) = x + 1 \; \mathrm{mod} \; 16$ over $\X$. (See caption for details.)
%



Evaluating \cref{eq:mainresult} for $\auxstates =4$, $|\X|=|\img(f)| = 16$, $\cycl(f) = 1$, and $\fix(f)=0$ gives
\begin{equation}
\Ctime[4]{f} = 5 + \extrabit[4] \,.
\end{equation}
Thus, the above construction of 5 idempotents achieves the minimal time cost for $4$ hidden states, and $\extrabit[4]=0$.
\end{example}
\noindent
See \cite{saito_products_1989} for details on how to decompose more complicated
functions into products of idempotent functions.


\subsection*{Visible states that are coarse-grained macrostates}
\label{sec:coarsegraining}

Our analysis above concerns scenarios where the full set of  
states is the union of the set of visible states
with the (disjoint) set of hidden states. However, in many real-world physical computers, $f$ is carried out over a set of macrostates that coarse-grain an underlying set of microstates. We call such macrostates ``logical states'' (logical states are sometimes called the states of the ``information bearing degrees of freedom''~\cite{bennett2003notes}). The map over the logical states, as specified by $f$, is 
induced by a master equation evolving over the underlying set of microstates. In such scenarios, we cannot express the full state space as the disjoint union of the logical states with some other ``hidden'' states, since the logical states are macrostates. 
This means that such scenarios can{not} be immediately analyzed with our framework.

However, as shown in the Methods, 
we can generalize our framework to include such maps  
carried out over logical macrostates, in such a way that scenarios involving
disjoint unions of visible and hidden states are just a special case.
It turns out that the results of the previous sections
apply without any modification, so long as we identify ``the number
of hidden states'' in those results with the difference between the number of microstates and the number of macrostates.

\begin{example} 
Suppose we have two
quantum dots, each with two possible states, written as $u$ and $w$, respectively, 
that evolve jointly according to a CTMC~\cite{diana2013finite}.
In this scenario the set of microstates  is the set of all four pairs $(u, w)$.

Suppose further that we identify a logical bit 
with the value of $u$.
Then a CTMC over $(u, w)$ will flip the value of the visible
state in two (hidden) timesteps if it implements the following sequence of two
idempotent functions:
\begin{enumerate}
\item $\{(0, 0), (0, 1)\} \mapsto (0, 0)$; $\{(1, 0), (1, 1)\} \mapsto (1, 1)$
\item $\{(0, 0), (1, 0)\} \mapsto (1, 0)$; $\{(1, 1), (0, 1)\} \mapsto (0, 1)$
\end{enumerate}
\noindent
Since there are four microstates and two logical states (given by
the value $u$), this means there are two ``hidden states''. Thus, applying
\cref{thm:auxstates}, with the appropriate change to how $\auxstates$ is defined,  
we conclude that no master equation can implement the bit flip using less than two timesteps.   This minimal time cost is in fact achieved by the construction above.
\label{ex:coarse_graining}
\end{example}

\section*{Discussion}

Many single-valued functions from initial to final states cannot be 
realized by master equation dynamics, even using time-inhomogeneous 
master equations. In this paper we show that any single-valued function  $f$ 
over a set of ``visible'' states $\X$ can be implemented, to arbitrary accuracy---if additional ``hidden'' states not in $\X$ are 
coupled to $\X$ by the underlying master equation. We refer to the minimal
number of hidden states needed to implement $f$ as the ``space cost'' of implementing $f$.
%
%
In addition, we show that given any function $f$ and number of 
available hidden states $\auxstates$, there is 
an associated minimal number of timesteps that are needed by any master equation
to implement $f$, where we define a ``timestep'' as a time interval
in which the set of allowed transitions between states does not 
change. We refer to this minimal number of timesteps
as the {time(step) cost of $f$ for $\auxstates$ hidden states}.

In this paper we derive a simple expression for the tradeoff between the space cost and the
time cost of any function $f$, a tradeoff
which depends on certain algebraic properties of $f$. 

We also analyze a generalization of our framework which encompasses scenarios
in which visible states are taken to be coarse-grained ``logical'' macrostates which carry out the desired input-output map, while the hidden states are a subset of the microstates over which the actual master equation unfolds. We show that all of our results regarding space and time costs still apply in this more 
general setting.

Interestingly, in standard treatments of the thermodynamics of computation, invertible functions can be carried out for free (i.e., while generating no heat), whereas many-to-one maps are viewed as costly. Moreover, noisy (i.e., non-single-valued) stochastic matrices can have lower thermodynamic cost than invertible single-valued ones, in the sense that the minimal free energy required to implement them can actually be negative~\cite{maroney2009generalizing,wolpert2015extending,wolpert_baez_entropy.2016}. In contrast, when considering the number of hidden states
required to implement a computation, it is many-to-one maps that are free, and single-valued invertible ones that are costly.  Furthermore, as shown in our companion paper~\cite{qepaper},
noisy maps may require more hidden states to implement than single-valued ones.
Thus, the relative benefits of many-to-one, invertible, and noisy maps are exactly opposite
when considering thermodynamic costs versus space and time costs.

The results derived in this paper are independent of considerations 
like whether detailed balance holds, how many thermal reservoirs the system is connected to, the amount of entropy production incurred by the stochastic process, etc. 
Nonetheless, in Supplementary Note 2, we show by construction that one can implement any $f$ using the minimal number of hidden states and timesteps using a master equation that (1) obeys detailed balance, (2) evolves probabilities in a continuous manner, and (3) has vanishing entropy
production, i.e., is thermodynamically reversible. The latter two properties are satisfied when the equilibrium distribution of the master equation at $t=0$ (determined by the choice of $q$ in the construction in Supplementary Note 2) coincides with the initial distribution over states (this and related issues are studied further in \cite{diss_wrong_prior_published_2017}).

This demonstrates that the implementation costs we consider are novel, and independent from energetic costs like heat and work that are traditionally studied in thermodynamics of computation. 
Indeed, while our analysis is driven by physical motivations,
it applies broadly to any field in which master equation dynamics play an important role.  


Our analysis suggests several important directions for future work:
\begin{enumerate}

\item Here,
we focused on tradeoffs involved in implementing single-valued functions, but typical real-world digital devices
cannot achieve perfect accuracy --- they will always have {some} noise. An important line for future work is to extend
our analysis to the space and timesteps tradeoffs for the case of arbitrary $P$, including non-single-valued maps. 
Some preliminary work related to this issue is presented in~\cite{qepaper}, where we 
present bounds (not exact equalities) on the space cost of arbitrary stochastic matrices. 
As discussed there, those space cost bounds
have some implications for bounds (again, inexact) on the time cost.

An associated goal is to
analyze the tradeoffs for implementing a given $f$ up to some accuracy $\epsilon$.
In this setting, a quantity of fundamental interest may be the maximal size $E_\text{max}$ of allowed energy barriers, which will determine how small entries of the rate matrix can be made. 
%
In particular, it is of interest to investigate the coupled tradeoffs  
between space cost, time cost (appropriately generalized), 
$\epsilon$, and $E_\text{max}$, and show how these 
reduce to a two-way tradeoff between space cost and time cost in the appropriate limit. (The analysis done here corresponds to the case where $\epsilon=0$ and $E_\text{max} = \infty$.)

\item Our results quantify the space-time tradeoff under the ``best-case'' scenario, where there are no restrictions on the dynamical processes available to an engineer who is constructing a system to carry out some map.  In particular, we assume that a system's dynamics can be sufficiently finely controlled so as to produce any desired idempotent function.   In real world situations, however, it is likely that the set of idempotent functions that can be 
engineered into a system
will be a tiny fraction of the total number possible, 
$\sum_{i=1}^{|\X|} \binom{|\X|}{i} i^{|\X|-i}$~\cite{harris1967number}. (This 
already exceeds a trillion if there are just 4 bits, so that $|\X| = 16$.) We perform an initial exploration of the consequences of such restrictions in Supplementary Note 6, but there is significant scope for future study of related tradeoffs.

\item Future work will also involve extending our framework to evaluate space and timestep tradeoffs for functions over infinite state spaces, in particular, to extend our results to Turing machines. (See \cite{qepaper} for preliminary analysis of the space costs of implementing noisy stochastic matrices over countably infinite spaces.)

\end{enumerate}


\section*{Methods}

Our proofs are fully constructive. At a high level, the construction 
can be summarized as follows:

\begin{enumerate}[(1)]

\item Adapting an existing result in semigroup theory~\cite{saito_products_1989}, we find the minimal (length) sequence of idempotent functions on a state space $\Y$ ($|\Y| = |\X| +k$) whose composition equals $f$ when restricted to $\X \subseteq \Y$.

\item We show that any idempotent function is one-step, by explicitly specifying (see  Supplementary Note 2) rate matrices and a limiting procedure for limit-embedding any idempotent function.  Thus, the length of the minimal sequence of idempotent functions whose composition implements $f$ with $\auxstates$ hidden states, as found in step (1), is an upper bound on $\Ctime{f}$. 

\item We show that if a CTMC implements $f$ with $\auxstates$ hidden states and $\auxsteps$ timesteps, then there must exist $\auxsteps$ idempotent functions whose composition implements $f$ with $\auxstates$ hidden states. Together with step (1) and (2), this means that $\Ctime{f}$ is exactly equal to the minimal number of idempotents whose composition implements $f$ with $\auxstates$ hidden states. 

\item Therefore, by chaining together the CTMCs implementing the idempotent functions in the decomposition we found in step (1), we construct a CTMC that implements $f$ while achieving our space and timestep bounds.

\end{enumerate}
\noindent
The rest of this section presents the details.

\subsection*{Time cost and idempotent functions}
\label{sec:methods}

Although our definitions apply to any stochastic matrix $\map$, our results all concern
$0/1$-valued stochastic matrices representing single-valued functions $f$. This is because 
there is a special relationship between one-step matrices that represent single-valued
functions and idempotent
functions, a relationship that in turn allows us to apply a result from semigroup theory to calculate time cost --- but only of single-valued functions.

We begin with the following, which is proved in Supplementary Note 2.

\begin{restatable}{theorem}{idempotentisstable}
Any idempotent function over a finite $\X$ is one-step.
\label{lemma:mto_implementable}
\end{restatable}

\cref{lemma:mto_implementable}  means that we can get an upper bound on the time cost of
a single-valued matrix $\map$ over a finite $\Y$ by finding the minimal number
of idempotent functions that equals $\map$. 
It turns out that this bound is tight, as proved in Supplementary Note 4:

\begin{restatable}{lemma}{towershanoi}
Suppose the stochastic matrix $\map$ over $\Y \supseteq \X$ has time cost $\auxsteps$ and
the restriction of $\map$ to $\X$ is a function $f : \X \to \X$.
Then there is a product of $\auxsteps$ idempotent functions over $\X$ whose restriction to $\X$ equals $f$.
\label{lem:towers_hanoi}
\end{restatable}

\noindent By combining these results, we simplify the calculation of the time cost of a function $f$ to the problem of finding a minimal set of idempotent functions whose product is $f$:

\begin{restatable}{corollary}{centralresult}
The time cost of any function $f$ with $\auxstates$ hidden states is the minimal
number of idempotents over $\Y = \X \cup \{1, \ldots, \auxstates\}$ such that
the product of those idempotents equals $f$ when restricted to $\X$. 
\label{thm:central_result}
\end{restatable}

%

Idempotent functions have been extensively studied in semigroup theory~\cite{erdos_products_1967,howie_subsemigroup_1966,howie_gravity_1979,howie_products_1984,saito_products_1989}. 
\cref{thm:central_result} allows us to exploit results from those studies to calculate the time cost (to within 1) for any function. In particular, we will use the following \namecref{thm:saito}, proved in~\cite{saito_products_1989} in an analysis of different issues:

\begin{theorem}
Let $f: \X \to \X$ be non-invertible. Then
\begin{equation}
\Ctime[0]{f} = \left\lceil\frac{|\X| + \cycl (f) - \fix (f)}{|\X| - |\img (f) |}\right\rceil + \extrabit[0] \,.
\end{equation}
where $\extrabit[0]$ equals either zero or one.
\label{thm:saito}
\end{theorem}
\noindent The expression for 
$\extrabit[0]$ is not easy to calculate, though some sufficient conditions for $\extrabit[0]=0$ are known~\cite{saito_products_1989}.

\subsection*{Proofs of our main results}
\label{subsec:main-proofs}

\setcounter{theorem}{0}
\begin{theorem}
Let $f : \X \to \X$. 
For any number of hidden states $\auxstates > 0$, the time cost is
\begin{multline}
\Ctime{f}  = \\
\left\lceil\frac{\auxstates + |\X| + \max \! \big[\!\cycl (f) - \auxstates, 0\big] - \fix (f) }{\auxstates + |\X| - |\img (f) |}\right\rceil + \extrabit
\end{multline}
where $\extrabit$ equals either zero or one.
\label{thm:auxstates}
\end{theorem}
\setcounter{theorem}{7}

\begin{proof}
	Let $\Y = \X \cup Z$ where $\Z \cap \X = \varnothing$ and $|\Z| = \auxstates$. By definition $\Ctime{f}$ is the minimum of $\Ctime[0]{g}$ over all non-invertible functions $g: \Y \to \Y$ that equal $f$ when restricted to $\X$. Moreover, by  \cref{thm:saito},
	\begin{equation}
	\Ctime[0]{g} = \left\lceil\frac{|\X| + \cycl (g) - \fix (g)}{|\X| - |\img (g) |}\right\rceil + \extrabit[0][g]
	\end{equation}

	Due to the constraint that $g(x) = f(x)$ for all $x \in \X$, our
	problem is to determine the optimal behavior of $g$ over $\Z$. For any fixed $|\img(g)|$, 
	this means finding the $g$ that minimizes $\cycl(g) -\fix(g)$. Since $\img (f) \subseteq \X$, the
	constraint tells us that there are no cyclic orbits of $g$ that include both elements of $\X$ and elements of $\Z$. 
	So all cyclic orbits of $g$ either stay wholly within $\Z$ or wholly within $\X$. 
	Moreover changing $g$ so that all elements of a cyclic orbit $\Omega$ lying wholly in 
	$\Z$ become fixed points of $g$ does not violate the constraint and reduces the time cost.
	Therefore under the optimal $g$, all $z \in \Z$ must either
	be fixed points or get mapped into $f(\X)$. 
	
	Our problem then reduces to determining precisely where $g$ should map those elements it sends 
	into $f(\X)$. To determine this, note that $g$ might map an element of $\Z$ into an $x$ that lies in a cyclic orbit of $f$, $\Omega$. If
	that happens, $\Omega$ will not be a cyclic orbit of $g$ --- and so the time cost will be reduced.
	Thus, to ensure that $\cycl(g)$ is minimal, we can assume that
	all elements of $\Z$ that are not fixed points of $g$ get mapped into $\img(f)$, with as many as possible
	being mapped into cyclic orbits of $f$.
	
	Suppose $g$ sends $m \leq \auxstates$ of the hidden states into the image of $f$, where each can be used to ``destroy'' a cyclic orbit of $f$ (until there are none left, if possible). 
	The remaining $\auxstates-m$ hidden states are fixed points of $g$. 
	Moreover, since $g(\X) = \img(f)$, 
	\begin{equation}|\img(g)| = |\img(f)| + \auxstates - m \,.\end{equation}
	So using \cref{thm:saito}, 
	\begin{multline}
	\Ctime[0]{g} = \\
	 \left\lceil\frac{m + |\X| + \max \! \big[ \! \cycl (f) - m, 0 \big]- \fix(f)}{m + |\X| - |\img(f) |}
	\right\rceil \!+\! \extrabit[0][g].
	\end{multline}
	
	\noindent The  quantity inside the ceiling function is minimized if $m$ is as large as possible,
	which establishes the result once we take $\extrabit[\auxstates][f] := \extrabit[0][g]$ for the $g$ which has $m=k$ and smallest $\extrabit[0][g]$. 
\end{proof}


\begin{corollary}
\label{cor:approx}
For any $f$ and number of hidden states $k$, 
\begin{align}
\Ctime{f} \approx \frac{|\X| + \cycl(f) - \fix(f)}{\auxstates + |\X| - |\img(f)|} + 1 \,.
\label{eq:approx1}
\end{align}
to within 2 timesteps.
\end{corollary}
\begin{proof}
Whenever $\auxstates \le \cycl(f)$, the approximation of  \cref{eq:approx1} holds up to accuracy of 1 timestep, since the $+1$ term accounts for error due to both the ceiling function and the term $\extrabit \in \{0,1\}$.
The equivalent approximation for $\auxstates > \cycl(f)$ is
\begin{align}
\frac{\auxstates + |\X| - \fix(f)}{\auxstates + |\X| - |\img(f)|} + 1 \,,
\label{eq:approx2}
\end{align}
and also holds up to accuracy of 1 timestep.
However, when $\auxstates > \cycl(f)$, \cref{eq:approx2} will never be more than 1 greater than \cref{eq:approx1}. To see why, note that \cref{eq:approx1} subtracted from \cref{eq:approx2} gives
\begin{align}
\frac{\auxstates - \cycl(f)}{\auxstates + |\X| - |\img(f)|}\,.
\label{eq:approxdiff}
\end{align}
For $\auxstates > \cycl(f)$, this quantity is bigger than 0.  At the same time, \cref{eq:approxdiff} is always smaller than 1, since the numerator is smaller than the denominator (observe that $|\X|-|\img(f)|\ge 0$).  
\end{proof}

\begin{corollary}
\label{cor:worstcase}
For any $f : \X \to \X$, 
\begin{align}
\Ctime{f} \le \frac{1.5 \times |\X|}{k} + 3 \,.
\end{align}
\end{corollary}
\begin{proof}
First, assume $|\X|$ is even and consider some function $f^*$ which has $f^*(f^*(x))=x$ and $f^*(x) \ne x$ for all $x\in \X$. 
One can verify that for this $f$, $\cycl(f)=|\X|/2$, $\fix(f) = 0$, and $|\img(f)| = |\X|$, and that these values maximize the approximation to the time cost given by \cref{cor:approx}. This approximation is accurate to within 2 timesteps, which implies the bound
\begin{align}
\Ctime{f} \le \frac{1.5 \times |\X|}{k} + 3 \,.
\end{align}
If $|\X|$ is odd, the maximum number of cyclic orbits is $(|\X|-1)/2$, so the above upper bound can be tightened by $1/(2k)$.
\end{proof}

\begin{corollary}
Let $\tau > 3$ and define
\begin{align}
\auxstates^{*} &:=  \left\lceil \frac{ \cycl(f)  - |\X|(\tau - 3) - \fix(f) }{(\tau - 2)} \right\rceil  + |\!\img(f)|  \nonumber \\
\auxstates^{**} &:=   \left\lceil \frac{  |\!\img(f)|(\tau - 2) - \fix(f) }{(\tau - 3)}  \right\rceil - |\X| \;.
\end{align}
We can implement $f$ in $\tau$ timesteps if we have at least $\auxstates$ hidden states,
where
\begin{equation} 
\auxstates = 
\begin{cases}  
\max [ \auxstates^{*}, 0] & \text{ if $\auxstates^{*} < \cycl(f)$} \\ 
\max [\auxstates^{**}, 0] & \text{ otherwise.} 
\end{cases} 
\end{equation}
\end{corollary}
\begin{proof}
	Since $\extrabit$ is always 0 or 1, by \cref{thm:auxstates} we know that we can
	implement $f$ if $\tau$ and $\auxstates$ obey
	\ba
	\tau &\ge  \left\lceil\frac{\auxstates + |\X| + \max \! \big[\!\cycl (f) - \auxstates, 0\big] - \fix (f) }{\auxstates + |\X| - |\img (f) |}\right\rceil + 1  \,.
	\ea
	This inequality will hold if
	\ba
	\tau &\ge  \frac{\auxstates + |\X| + \max \! \big[\!\cycl (f) - \auxstates, 0\big] - \fix (f) }{\auxstates + |\X| - |\img (f) |} + 2 \,.
	\ea
	The RHS 
	is non-increasing in $\auxstates$.  
	So we can implement $f$ in $\tau$ timesteps, as desired, if 
	$\auxstates$ is the smallest integer that obeys the inequality.
	
	First hypothesize that the smallest such $n$ is less than $\cycl(f)$. In this case
	$\max \! \big[\!\cycl (f) - \auxstates, 0\big] = \cycl(f) - \auxstates$. So our bound becomes
	\ba
	\tau &\ge  \frac{\auxstates + |\X| +\cycl (f) - \auxstates - \fix (f) }{\auxstates + |\X| - |\img (f) |} + 2 \,,
	\ea
	which is saturated if
	\ba
	\auxstates= 
	\frac{|\X|(3 - \tau) - \fix(f) + \cycl(f)}{(\tau - 2) } + |\img(f)| \,.
	\ea

	If instead the least $\auxstates$ that obeys our inequality is greater than or equal to $\cycl(f)$, then our bound becomes
	\ba
	\tau &\ge  \frac{\auxstates + |\X| - \fix (f) }{\auxstates + |\X| - |\img (f) |} + 2  \,,
	\ea
	which is saturated if
	\ba
	\auxstates = \frac{ |\!\img(f)|(\tau - 2) - \fix(f)}{\tau - 3} - |\X|  \,.
	\ea

	The fact that $\auxstates$ must be a nonnegative integer completes the proof.
\end{proof}

%
%

\begin{corollary}
Any $f$ can be implemented in two timesteps, as long as $\vert \img(f) \vert$ hidden states are available.
\label{ex:maxstates}
\end{corollary}
\begin{proof}
	Consider an implementation of $f$ when $\auxstates = \vert \img(f) \vert$ hidden states are available. Index the states in $\Y$ using $1,\dots,\vert \X \vert$ for the states in $\X$ and $\vert \X \vert + 1,\dots, \vert \X \vert + \auxstates$ for the hidden states. The function $f$ can then be implemented as a product of two idempotents:
	\begin{enumerate}
		\item In the first step, for each $x \in \X$, both $x$ and $\auxstates+f(x)$ are mapped to $k+f(x)$;
		\item In the second step, for each $x' \in \img(f)$, both $x'$ and $\auxstates+x'$ are mapped to $x'$.
	\end{enumerate}
\end{proof}

\begin{corollary}
If $f: \X \to \X$ is a cyclic permutation with no fixed points 
and there is one hidden state available, then the time cost is $|\X| + 1$. 
\label{cor:cyclicorbit}
\end{corollary}
\begin{proof}
	\cref{thm:auxstates} tells us that the time cost of $f$ is $|\X| + 1$ or $|\X| + 2$.  To show that it is in fact $|\X|+1$, write the states of $\X$ as
	$\{1, 2, \ldots, |\X|\}$, with the single hidden state written as $|\X| + 1$. Assume without loss of generality that the states are numbered so that $f(i) = i+1 \; \mathrm{mod} \; |\X|$. Then have the first idempotent function send 
	$\{|\X|, |\X| + 1\} \mapsto |\X| + 1$ (leaving all other states fixed), the second function send  
	$\{|\X| - 1, |\X|\} \mapsto |\X| $ (leaving all other states fixed), etc., up to the $|\X|$'th idempotent function, which sends
	$\{1, 2\} \mapsto 2$ (leaving all other states fixed). Then have the last idempotent function send 
	$\{1, |\X| + 1\} \mapsto 1$ (leaving all other states fixed). It is easy to verify that this sequence of $|\X| +1 $ idempotent
	functions performs the cyclic orbit, as claimed.
\end{proof}
\noindent It is straightforward to use the proof technique of \cref{cor:cyclicorbit} to show that, in \cref{thm:auxstates}, $b(f, 1) = 0$ for any invertible $f$.

\subsection*{Extension to allow visible states to be coarse-grained macrostates}
\label{sec:coarse-graining}

If the visible states are identified with a
set of macrostates given by coarse-graining an underlying set of microstates, then
the framework introduced above, where $\X \subseteq \Y$,
does not directly apply. It turns out though that we can generalize
that framework to apply to such scenarios as well. To show
how we start with the following definition:

\newcommand{\Xm}{\mathcal{Z}}
\newcommand{\fm}{\hat{f}}
\newcommand{\nm}{n}
\newcommand{\om}{\omega}

\begin{definition}
\label{def:cg}
A function $\fm:\Xm \to \Xm$
{can be implemented with $\nm$ microstates and $\ell$ timesteps}
if and only if there exists a set $\Y$ with $\nm$ states and a partial
function $g : \Y \to \Xm$ such that
\betight
\item $\mathrm{img}(g)=\Xm$,\label{enu:imgcond}
\item there exists a stochastic matrix $M$ over $\Y$ which is a product of $\ell$ one-step matrices,
\item for all $i \in \dom(g)$, $\sum_{j \in g^{-1}(\fm(g(i))} M_{ji} = 1$.
\label{cond:sum}
\eetight
The minimal number $\nm$ such that $\fm$ can be
implemented with $n$ microstates (for some associated $g$ and $M$, and any number of timesteps)
we call the {microspace cost} of $\fm$.
\end{definition}

Note that we allow the coarse-graining function to be partially specified,
meaning that some microstates may have an undefined corresponding macrostate.
Nonetheless, condition \ref{enu:imgcond} in \cref{def:cg}
provides that each macrostate is mapped to by at least one microstate.
An example of \cref{def:cg} is given by the class of scenarios
analyzed in the previous sections, in which $\Xm = \X\subseteq \Y$,
$g(x)=x$ for all $x\in \X$ and is undefined otherwise, and
the elements $\Y\backslash \X$ are referred to as hidden states. Note, however, that
in \cref{def:cg}, we specify a number of microstates,
rather than a number of hidden states. As illustrated in \cref{ex:coarse_graining},
this flexibility allows us to consider scenarios in which each $z\in \Xm$
is not a single element of the full space $\Y$, but rather a coarse-grained
macrostate of $\Y$.

\begin{definition}
Let $\fm$ be a single-valued function over $\Xm$ that can be implemented 
with $\nm$ microstates.
Then we say that the {(hidden) time(step) cost} of $\fm$ with $\nm$ microstates
is the minimal number $\auxsteps$ such that
$\fm$ can be implemented with $\nm$ microstates.
\label{def:extended_switch}
\end{definition}

\noindent The minimization in \cref{def:extended_switch} is implicitly over 
the set of partial macrostates, the matrix $M$, and the function ${g}$.

The proof of the following \namecref{thm:macromicro} is left for the Supplementary Information.
\begin{restatable}{theorem}{thmmacromicro}
Assume $\fm : \Xm \to \Xm$ can be implemented with $\nm$ microstates and $\auxsteps$ timesteps. Then there is a stochastic matrix $W$ over 
a set of $\nm$ states $\Y$, a subset $\X\subseteq \Y$ with $|\X| = |\Xm|$, 
and a one-to-one mapping $\om : \Xm \to \X$ such that
\betight
	\item $W$ is a product of $\auxsteps$ one-step matrices
	\item The restriction of $W$ to $\X$ carries out the function $f(x) := \om(\fm(\om^{-1}(x))$
\eetight
\label{thm:macromicro}
\end{restatable}

We are finally ready to prove the equivalence between time cost as defined in previous sections, and time cost for computations over coarse-grained spaces.

\begin{corollary}
Consider a system with microstate space $\Y$. The hidden time cost of a function $\fm$ over a coarse-grained space $\Xm$ with $n$ microstates equals the
hidden time cost of $\fm$ (up to a one-to-one mapping between $\Xm$ and $\X \subseteq \Y$) with $n-|\Xm|$ hidden states.
\label{coroll:coarse-graining}
\end{corollary}
\begin{proof}
Let $\auxsteps$ indicate the time cost of $\fm : \Xm \to \Xm$ with $\nm$ microstates, and let $M$ be a stochastic matrix that achieves this (microstates-based) time cost.  Similarly, let $\auxsteps'$ indicate the time cost of carrying out $\hat{f}$ over $\X \subseteq \Y$ (up to a one-to-one mapping between $\Xm$ and $\X$, which we call $\omega : \Z \to \X$) with $n-|\Xm|$ hidden states, and let $M'$ be a stochastic matrix that achieves this (hidden-states-based) time cost. We prove that $\auxsteps = \auxsteps'$ by proving the two inequalities, $\auxsteps \le \auxsteps'$ and $\auxsteps' \le \auxsteps$.

By \cref{thm:macromicro}, it must be that there exists an implementation of $\fm$ over $\X$ with $n-|\Xm|$ hidden states and $\auxsteps$ timesteps. Thus, $\auxsteps' \le \auxsteps$.  We can also show that $\auxsteps \le \auxsteps'$.  To do so, define the coarse-graining function $g(x) := \omega^{-1}(x)$ for all $x \in \img(\omega)$, and $g(x)$ undefined for all $x \not\in\img(\omega)$. It is easy to verify that $M$ and $g$ satisfies the conditions of \cref{def:cg} with $n$ microstates and $\auxsteps'$ timesteps. Thus, $\auxsteps \le \auxsteps'$.
\end{proof}

\section*{Data availability statement}

No datasets were generated or analysed during the current study.

\begin{acknowledgments}
We would like to thank the Santa Fe Institute for helping to support this research. 
This paper was made possible through the support of Grant No.
FQXi-RFP-1622 from the FQXi foundation, and Grant No. CHE-1648973 from the U.S. National Science Foundation. 
\end{acknowledgments}

\section*{Competing Interests}

The authors declare no competing interests.

\section*{Contributions}

DHW came up with the project; the research was done by AK, JAO, DHW;
the writing was done by AK, JAO, DHW.

%


\clearpage
\setcounter{theorem}{14}

\renewcommand\appendixname{Supplementary Note}
\appendix
\renewcommand\thesection{\arabic{section}}

\renewcommand{\theequation}{\arabic{equation}}
\setcounter{definition}{9}

\section{Explicit demonstration that bit erasure is a one-step function}

\label{app:ex:1}
In the model of bit erasure described in~\cite{diana2013finite} a classical bit is stored in a quantum dot, which can be either empty (state 0) or filled with an electron (state 1).  The dot is brought into contact with a metallic lead at temperature $\temp$ which can transfer an electron to/from the dot. The propensity of the lead to give an electron is set by its chemical potential, indicated by $\mu(t)$ at time $t$. The energy of an electron in the dot is indicated by $E(t)$.

Let $p(t)$ indicate the two-dimensional vector of probabilities at time $t$, with $p_0(t)$ and $p_1(t)$ being the probability of an empty and full dot, respectively.  These probabilities evolve according to 
a rate matrix~\cite{diana2013finite}:
\begin{align}
\dot{p}(t) = C
\begin{bmatrix}
 - w(t) & 1-w(t) \\
  w(t) &  -(1-w(t))
\end{bmatrix} p(t)
\label{eq:quantum-dot-deriv}
\end{align}
where $C$ sets the timescale of the exchange of electrons between the dot and the lead and 
$w(t)$ is the Fermi distribution of the lead,
\begin{equation}
w(t) = \left[\exp((E(t)-\mu(t))/k_B \temp)+1\right]^{-1}\,.
\end{equation}

Using \eqref{eq:quantum-dot-deriv} and conservation of probability (i.e., $p_0(t) + p_1(t) = 1$), we can write
\begin{equation}\label{master}
\dot{p_1}(t) = C (w(t) - p_1(t)) \,,
\end{equation}
so $p_1(t)=w(t)$ is the stationary state at time $t$.

Suppose that the chemical potential $\mu(t)$ and electron energy $E(t)$ are chosen in such a way that $w(t) = (1-t) q + t \delta$ for some constants $q$ and $\delta$. 
In this case, \eqref{master} can be explicitly solved for $p_1$,
\begin{equation}
p_1(t) = w(t) + e^{-Ct}\left(p_1(0)-q\right)+ C^{-1}(q-\delta)\left(1-e^{-Ct}\right).
\end{equation}
In the limit where $C \to \infty$ and $\delta \to 0$, we have \begin{equation}p_1(t) = w(t) = (1-t)q \,,\end{equation}
which corresponds to the transition matrix
\begin{align}
T(0, t) =
\begin{bmatrix}
1-(1-t)q & 1-(1-t)q \\
(1-t) q & (1-t) q
\end{bmatrix}\,.
\end{align}
Note that 
$T(0,1) = \left( \begin{smallmatrix} 1&1\\ 0&0 \end{smallmatrix} \right)$, so the process implements bit erasure.  By \cref{lem:prod-T}, it must also be that $T(t,1) = \left( \begin{smallmatrix} 1&1\\ 0&0 \end{smallmatrix} \right)$.  We note that $T(0,t)$ and $T(t,1)$ are continuous in $t$ and have a constant set of allowed transitions over $t \in (0,1)$, which establishes that bit erasure is one-step.

\section{Properties of master equations that implement idempotent functions in one timestep}
\label{app:singletimestep}

We begin by proving that any idempotent function over a finite $\X$ is one-step,
\cref{lemma:mto_implementable}.
Let $f$ be an idempotent function, and let $\map_{ij} = \delta(i,f(j))$ be the corresponding stochastic matrix.  
We use an explicit construction to show that there exists a sequence of CTMCs $\{ T^{(n)} : n=1,2,\dots \}$  which 
obey the conditions of \cref{def:singletimestep}. 

First, choose any arbitrary probability distribution $q$ over $\X$, and let $q_i$ indicate the probability of state $i$. Define
\begin{align}
\tilde{q}_i := \begin{cases}
q_i/\sum_{j:f(j)=f(i)} q_j & \text{if $\sum_{j:f(j)=f(i)} q_j > 0$}\\
0 & \text{otherwise}
\end{cases}
\label{eq:qtilde}
\end{align}
Each $\tilde{q}_i$ is the `renormalized' probability within the block of states $\{ j : f(j) = f(i) \}$.  

Then, for all $i$,  define
\begin{equation}
w_i(t) = (1-t)\tilde{q}_i + t \delta(i, f(i)) \,.
\end{equation}
where $\delta(\cdot,\cdot)$ is the Kronecker delta function.

Then, define the rate matrix $Q^{(n)}(t)$ as 
\begin{align}
Q_{ij}^{(n)}(t) = \begin{cases}
n w_i(t) & \text{if $i\ne j$ and $f(i) = f(j)$} \\
n (w_i(t) - 1 ) & \text{if $i=j$} \\
 0 & \text{otherwise}
\end{cases}
\label{eq:ratematrix}
\end{align}
It can be verified that if $f$ is an idempotent function, then $Q^{(n)}(t)$ is a valid rate matrix (that is, $Q^{(n)}_{ij}(t)\ge 0$ for all $i,j$ and $\sum_i Q^{(n)}_{ij}(t) = 0$ for all $j$).

Next, for any $n\in \mathbb{N}$, define the CTMC $T^{(n)}(t,t')$ as the solution to the following differential equation,
\begin{align}
T^{(n)}_{ij}(t,t) & = \delta(i,j)  \label{eq:idcond} \\
\frac{d}{dt'} T^{(n)}_{ij}(t,t') & = \sum_{k} Q^{(n)}_{ik}(t') T^{(n)}_{kj}(t,t') \label{eq:derivcond}
\end{align}
We can simplify \cref{eq:derivcond} by using the definition of $Q^{(n)}(t)$. First note that no probability can ever flow from state $j$ to state $i$ if $f(i) \ne f(j)$, hence for such $i,j$, $T^{(n)}_{ij}(t,t') = 0$ always.  On the other hand, for $i,j$ where $f(i)=f(j)$, we 
can rewrite 
\begin{align}
& \frac{d}{dt'} T^{(n)}_{ij}(t,t') = \sum_{k} Q^{(n)}_{ik}(t') T^{(n)}_{kj}(t,t')  \nonumber \\
 & = n \left[  (w_i(t') - 1) T^{(n)}_{ij}(t,t') + \quad \mathclap{\sum_{k : k\ne i, f(k)=f(i)}} \quad w_i(t') T^{(n)}_{kj}(t,t') \right] \nonumber \\
 & = n \left(  (w_i(t')-1) T^{(n)}_{ij}(t,t') + w_i(t') \left( 1- T^{(n)}_{ij}(t,t') \right) \right) \nonumber \\
 & = n \left(   w_i(t') - T^{(n)}_{ij}(t,t') \right) \label{eq:simplified}
\end{align}

\cref{eq:simplified}, in combination with initial condition \cref{eq:idcond}, can be explicitly solved to give 
\begin{multline}
T^{(n)}_{ij}(t,t')  = w_i(t') + (\delta(i,j) - w_i(t))e^{-(t'-t)/n} \\
+  n^{-1}(\tilde{q}_i - \delta(i,f(i)))(1-e^{-(t'-t)/n}) \,.
\end{multline}
The $n\rightarrow \infty$ limit for $t' \ge t$ is 
\begin{align}
T_{ij}(t,t') & := \lim_{n\rightarrow \infty} T^{(n)}(t,t') \\
& =
\begin{cases}
w_i(t') & \text{if $t' > t$ and $f(i)=f(j)$} \\
0 & \text{if $t' > t$ and $f(i)\ne f(j)$}  \\
\delta(i,j) & \text{if $t=t'$}
\end{cases}
\label{eq:Tlimit}
\end{align}
As a particular case, for $t=0, t'=1$, we have
\begin{align}
T_{ij}(0,1) = \delta(i,f(j)) = P \,,
\label{eq:b7}
\end{align}
where we've used the fact that $w_i(1) = \delta(i,f(i))$.  

We have thus shown that $\{ T^{(n)} : n = 1,2, \dots \}$ is a limit-embedding of $P$,
as required for any one-step matrix. Next, the condition \cref{def:singletimestep}(1) on the sequence $\{ T^{(n)} : n = 1,2, \dots \}$ is met by inspection. 
In addition, since $w_i(t)$ is a continuous function of $t \in [0, 1]$, 
it follows both that $T(0, t)$ is a continuous function of $t$ for all $t \in (0, 1]$
and that $T(t, 1)$ is a continuous function of $t$ for all $t \in [0, 1)$. This
establishes that \cref{def:singletimestep}(2) holds.  Finally,
\cref{def:singletimestep}(3) holds by construction. 

Thus, all the conditions
given in \cref{def:singletimestep} concerning the limiting matrix $T(t,t')$ are
satisfied, which establishes the claim that $P$ is a one-step matrix. 
Note in particular that even though $\{ T^{(n)} : n = 1,2, \dots \}$ is defined
in terms of one particular initial distribution $q$, the associated transition
matrix $T(0, t)$ implements $P$ no matter what the initial distribution is.



It is worth highlighting three properties of the 
construction above. 
%

{First}, when $q$ equals the initial distribution $p(0)$, the 
function $p(t) = T(0,t)p(0)$ is a continuous function of $t$
for all $t \in [0, 1]$. To see this, first note
that since 
$T(0, t)$ is continuous for all $t \in (0, 1]$,
$T(0,t)p(0)$ is continuous for all $t \in (0, 1]$.
Moreover,
\begin{align}
\lim_{t\rightarrow 0^+} p_i(t) & = \lim_{t\rightarrow 0^+} \sum_j T_{ij}(0,t) p_j(0) \nonumber \\
& = \sum_{j:f(j)=f(i)} w_i(0) p_j(0) \nonumber \\
& = \sum_{j:f(j)=f(i)} \frac{q_i}{\sum_{j':f(j')=f(i)} q_j} p_j(0) = p_i(0)
\end{align}
Therefore $p(t)$ in fact is continuous for all $t \in [0, 1]$, as claimed.

{Second,} when $q=p(0)$, then the above construction results in
no (irreversible) entropy production. More precisely, stochastic thermodynamics provides a simple formula for  
the rate of entropy production incurred by a system evolving according to a master equation, while being coupled to a thermodynamic reservoir~\cite{esposito2012stochastic, esposito2010three}:
%
\begin{proposition}
Consider a CTMC with finite rates $Q(t)$ and let $p(t)$ be a
distribution of states at time $t$ of a system that evolves according to that CTMC.  
The {(irreversible) entropy production} rate at time $t$ is 
\begin{align}
\dot{\Sigma}(Q(t), p(t)) := \sum_{i,j} p_j(t) Q_{ij}(t) \ln\frac{p_{j}(t)Q_{ij}(t)}{p_{i}(t)Q_{ji}(t)} 
\label{eq:EPdef}
\end{align}
The {integrated entropy production} over $t\in[0,1]$ is 
\begin{align}
\label{TotalEP}
\Sigma(Q, p(0)) = \int_0^1 \dot{\Sigma}(Q(t), p(t)) \; dt \,.
\end{align}
\end{proposition}
\noindent Now consider the rate matrices $Q^{(n)}(t)$ defined in \cref{eq:ratematrix}.  Note that for all $t \in [0,1]$, these rate matrices have a fixed ``block structure'', in which transitions are allowed between states $i,j$ in the same block ($f(i)=f(j)$), but not allowed between states $i,j$ in different blocks ($f(i)\ne f(j)$).  It is straightforward to verify that for  block-structure rate matrices, one can rewrite \cref{eq:EPdef,TotalEP} as a weighted sum of entropy production terms arising from each block. In particular, letting $S_k = f^{-1}(k)$ be the preimage of $k$ under $f$, we can rewrite \cref{TotalEP} as
\begin{align}
\Sigma(Q^{(n)}, p(0)) = \sum_k p^k(0)\; \Sigma\left(Q^{(n)}_{[S_k]}, p_{[S_k]}(0)/p^k(0)\right) \,,
\end{align}
where $p^k(0) = \sum_{i\in S_k} p_i(0)$, $p_{[S_k]}$ is the restriction of the distribution $p$ to the states in $S_k$, and $Q^{(n)}_{[S_k]}$ uses the notation from \cref{def:restriction}. 
Then, each $Q^{(n)}_{[S_k]}$ is irreducible and (if $q = p(0)$) exactly follows the construction specified in the Appendix D of the companion paper~\cite{qepaper}. 
In that Appendix, we prove that
\begin{equation}
\lim_{n \rightarrow \infty} \Sigma\left(Q^{(n)}_{[S_k]}, p_{[S_k]}(0)/p^k(0)\right) = 0 \,.
\end{equation}
Thus, in the $n\rightarrow \infty$ limit, the integrated entropy production vanishes.

{Third}, we note that we can build a CTMC that implements a composition of idempotents
by ``gluing together" the CTMC corresponding to each idempotent in turn.
For example, suppose we wish to implement a map $h = f \circ g$,
where $f$ and $g$ are idempotents with corresponding stochastic matrices $P_1$, $P_2$. 
Write $Q^{(n)}_1$ and $Q^{(n)}_2$ for the 
rate matrices implementing $P_1$ and $P_2$ respectively (as in \cref{eq:ratematrix}).
Then, we can implement $h$ by taking the $n \rightarrow \infty$ limit of the rate matrices
\begin{eqnarray}
Q^{(n)}(t) = \begin{cases*}
                    Q^{(n)}_1(2t) & if  $t \in [0, \frac12]$  \\ 
                     Q^{(n)}_2(2t-1) & if $t \in (\frac12, 1]$
                 \end{cases*} \,.
\label{eq:composing_ctmcs}
\end{eqnarray}

\section{Transitivity condition on one-step matrices}
\label{app:thm3}

One particularly useful property of one-step matrices involves
a kind of transitivity of probability flow, formalized as follows:

\begin{definition}
A stochastic matrix $\map$ is {transitive} if for all triples of states
$\{i,j,k\}$ such that $P_{ji}>0$ and $P_{kj}>0$, it is also true that $P_{ki}>0$.
\end{definition}

In this Supplementary Note we show that one-step matrices are transitive. To
do this we start with a pair of simple lemmas. In all of them we take $\map$ to be a matrix that is limit-embeddable by $T$, and such that the limit $T(t,t'):=\lim_{n\rightarrow\infty}T^{(n)}(t,t')$ exists for all $t,t'\in[0,1]$.
\begin{lemma}
For any $t\in[0,1]$, $T(0,1)=T(t,1)T(0,t)$.
\label{lem:prod-T}
\end{lemma}
\begin{proof}
Note that any embeddable CTMC in the sequence $T^{(n)}$ obeys the Chapman-Kolmogorov
equations,
\begin{equation}
T^{(n)}(0,1)=T^{(n)}(t,1) T^{(n)}(0,t)
\end{equation}
Since the limit of a product is the product of limits, we can write
\begin{align}
T(0,1) & =\lim_{n\rightarrow\infty}T^{(n)}(0,1)\nonumber \\
 & =\lim_{n\rightarrow\infty}T^{(n)}(t,1) T^{(n)}(0,t)\nonumber \\
 & =\left(\lim_{n\rightarrow\infty}T^{(n)}(t,1)\right) \left(\lim_{n\rightarrow\infty}T^{(n)}(0,t)\right)\nonumber\\
 & =T(t,1) T(0,t)\,.
\end{align}
\end{proof}
\begin{lemma}
If $\map$ is one-step and $T_{ij}(0,1)>0$  for some pair of states $i$ and $j$, then $T_{ij}(0,t)>0$
for all $t\in(0,1)$.\label{theorem:constantprob}
\end{lemma}
\begin{proof}
If $T_{ij}(0,1)>0$, by continuity of $T(0, t)$ in $t$, there must
be a $t'\in(0,1)$ such that $T_{ij}(0,t')>0$. The claim follows from
the definition of a one-step matrix.
\end{proof}

\begin{restatable}{theorem}{thmtransitivity}
If $\map$ is one-step it is transitive.\label{theorem:transitivity}
\end{restatable}
\begin{proof}
Recall that $P=T(0,1)$, and consider any three states $i$, $j$,
and $k$ such that $T_{ji}(0,1)>0$ and $T_{kj}(0,1)>0$. Given that
$T_{kj}(0,1)>0$, by continuity of $T(t, 1)$ in $t$ there must be
a $t'\in(0,1)$ such that $T_{kj}(t',1)>0$. By \cref{theorem:constantprob},
given that $T_{ji}(0,1)>0$, $T_{ji}(0,t)>0$ for all $t\in(0,1)$.
Combining with \cref{lem:prod-T} gives
\begin{multline}
T_{ki}(0,1)  =\sum_{j'}T_{kj'}(t',1)T_{j'i}(0,t') \\
\ge T_{kj}(t',1) T_{ji}(0,t')>0\,.
\end{multline}
Thus, if $P_{ji}>0$ and $P_{kj}>0$, $P_{ki}>0$.
\end{proof}

\section{Calculating time cost using products of idempotent functions}
\label{app:lemma2}

For convenience, in this Supplementary Note we define the {{adjacency matrix}} of a matrix $K$ as
\begin{align}
\adj[K]_{ij}=\begin{cases}
1 & \text{if }K_{ij}>0\\
0 & \text{otherwise}
\end{cases}.
\end{align}
It can be verified that condition \ref{cond:constantadjacnecy} of \cref{def:singletimestep} (one-step matrix) is equivalent to stating that $\adj[T(0,t)]$ is constant over $t \in (0,1)$.

We also use $\adj[\tilde{L}]_{ij}\subseteq\adj[L]_{ij}$ to indicate that
$\adj[\tilde{L}]_{ij}=0$ whenever $\adj[L]_{ij}=0$, for all states $i,j$.
\begin{lemma}
For any one-step matrix $L$, there exists a one-step matrix $\tilde{L}$
which carries out an idempotent function and which has $\adj[\tilde{L}]\subseteq\adj[L]$.
\label{lem:idem-equiv}
\end{lemma}
\begin{proof}
Let $G$ be the graph that corresponds to $\adj[L]$. Since $L$ is
a stochastic matrix, every node in $G$ must have at least one outgoing
edge. Since the number of nodes is finite, this means that there must
be a path from every node to at least one node in a directed cycle.
Furthermore, since $L$ is one-step,  $G$ must be transitive (\cref{theorem:transitivity}). Thus, every
node must have at least one direct edge to a node in a cycle. Furthermore,
for any node in a cycle, there is a directed path from itself back
to itself. Since $G$ is transitive, any node in a cycle must therefore
have an edge to itself (self-loop).

Thus, any node in $G$ must either have a self-loop, or must be directly
connected to at least one other node with a self-loop. For each node
without a self-loop, let $v_{i}$ indicate any node that $i$
is connected to and which has a self-loop. Define the stochastic matrix
$\tilde{L}$ in the following manner: for any node $i$ and all $j$, let
$\tilde{L}_{ji}=\delta_{i,j}$ if $i$ has a self-loop, and let $\tilde{L}_{ji}=\delta_{j,v_{i}}$
if $i$ doesn't have a self-loop. By construction, $\adj[\tilde{L}]\subseteq\adj[L]$.
It is straightforward to check that $\tilde{L}$ is idempotent: every
$i$ with a self-loop is sent to itself no matter how many times $\tilde{L}$
is applied, and every $i$ without a self-loop is sent to $v_{i}$,
no matter how many times $\tilde{L}$ is applied.

$\tilde{L}$ is one-step by \cref{lemma:mto_implementable}. 
\end{proof}

\begin{lemma}
Consider two stochastic matrices $A$ and $\tilde{A}$ over $\Y$, each expressible as a product
of $n$ stochastic matrices,
\begin{align}
A =L^{(n)}L^{(n-1)}\dots L^{(1)} \qquad \tilde{A} =\tilde{L}^{(n)}\tilde{L}^{(n-1)}\dots\tilde{L}^{(1)}
\end{align}
If for all $i=1..n$, $\adj[\tilde{L}^{(i)}]\subseteq\adj[L^{(i)}]$,
then $\adj[\tilde{A}]\subseteq\adj[A]$.
\label{lem:prod-equiv}
\end{lemma}
\begin{proof}
Define the following partial products,
\begin{align}
A^{[k]} & =L^{(k)}L^{(k-1)}\dots L^{(1)}  \qquad \tilde{A}^{[k]} =\tilde{L}^{(k)}\tilde{L}^{(k-1)}\dots\tilde{L}^{(1)}
\end{align}
We prove the Lemma, i.e., that $\adj[\tilde{A}^{[n]}]\subseteq\adj[A^{[n]}]$,
by induction in $k$.

Observe that since ${A}^{[1]}=L^{(1)}$ and $\tilde{A}^{[1]}=\tilde{L}^{(1)}$,
by assumption $\adj[\tilde{A}^{[1]}]\subseteq\adj[A^{[1]}]$. Now
write 
\begin{equation}
A_{ij}^{[k]}=\sum_{l}L_{il}^{(k)}A_{lj}^{[k-1]}
\end{equation}
If $A_{ij}^{[k]}=0$, this means that $\forall l\in \Y$, $L_{il}^{(k)}=0$
and $A_{lj}^{[k-1]}=0$. But since $\adj[\tilde{L}^{(k)}]\subseteq\adj[L^{(k)}]$,
$L_{il}^{(k)}=0$ implies $\tilde{L}_{il}^{(k)}=0$; similarly, given
$\adj[\tilde{A}_{lj}^{[k-1]}]\subseteq\adj[A_{lj}^{[k-1]}]$, $A_{lj}^{[k-1]}=0$
implies $\tilde{A}_{lj}^{[k-1]}=0$. Thus, if $A_{ij}^{[k]}=0$, then
it must be that
\begin{equation}
\tilde{A}_{ij}^{[k]}=\sum_{l}\tilde{L}_{il}^{(k)}\tilde{A}_{lj}^{[k-1]}=0
\end{equation}
Therefore, if $\adj[\tilde{A}_{lj}^{[k-1]}]\subseteq\adj[A_{lj}^{[k-1]}]$
and $\adj[\tilde{L}^{(k)}]\subseteq\adj[L^{(k)}]$, then $\adj[\tilde{A}_{lj}^{[k]}]\subseteq\adj[A_{lj}^{[k]}]$.
\end{proof}

\setcounter{theorem}{4}
\begin{restatable}{lemma}{towershanoi}
Suppose the stochastic matrix $\map$ over $\Y \supseteq \X$ has time cost $\auxsteps$ and
the restriction of $\map$ to $\X$ is a function $f : \X \rightarrow \X$.
Then there is a product of $\auxsteps$ idempotent functions over $\X$ whose restriction to $\X$ equals $f$.
\label{lem:towers_hanoi}
\end{restatable}
\begin{proof}
By hypothesis we can write $\map = \stablematrix^{(\auxsteps)}\stablematrix^{(\auxsteps-1)}\dots \stablematrix^{(1)}$ where each $\stablematrix^{(i)}$ is one-step. 
By \cref{lem:idem-equiv}, for each $\stablematrix^{(i)}$ there is another one-step matrix $\tilde{\stablematrix}^{(i)}$ which carries out an idempotent function, and which has $\adj[\tilde{\stablematrix}^{(i)}] \subseteq \adj[\stablematrix^{(i)}]$.
By \cref{lem:prod-equiv}, the product of these idempotent functions, $\tilde{\map} = \tilde{\stablematrix}^{(\auxsteps)}\tilde{\stablematrix}^{(\auxsteps-1)}\dots \tilde{\stablematrix}^{(1)}$, obeys $\adj[\tilde{\map}] \subseteq \adj[\map]$.

The restriction of $\map$ to $\X$ implements the single-valued function $f: \X\rightarrow \X$, meaning that $\map_{ji} = \delta_{f(i),j}$ for all $i \in \X$. Therefore, it must be that $\tilde{\map}_{ji} = \delta_{f(i),j}$ for all $i \in \X$, since otherwise $\tilde{\map}$ would have a nonzero entry in a location where $\map$ has a 0 entry (contradicting $\adj[\tilde{\map}] \subseteq \adj[\map]$).  Therefore, the restriction of $\tilde{\map}$ to $\X$ must  equal $f$.
\end{proof}

As an aside, \cref{lem:towers_hanoi} tells us that if $\X = \Y$, and
$\map$ is single-valued and one-step (so $\auxsteps = 1$), then $\map$ must be an idempotent function.

\section{Time cost where visible states are macrostates}

\setcounter{theorem}{12}
\begin{restatable}{theorem}{thmmacromicro}
Assume $\fm : \Xm \rightarrow \Xm$ can be implemented with $\nm$ microstates and $\auxsteps$ timesteps. Then there is a stochastic matrix $W$ over 
a set of $\nm$ states $\Y$, a subset $\X\subseteq \Y$ with $|\X| = |\Xm|$, 
and a one-to-one mapping $\om : \Xm \rightarrow \X$ such that
\betight
	\item $W$ is a product of $\auxsteps$ one-step matrices
	\item The restriction of $W$ to $\X$ carries out the function $f(x) := \om(\fm(\om^{-1}(x))$
\eetight
\label{thm:macromicro}
\end{restatable}
\begin{proof}
Assume $\hat{f}$ is implemented with $n$ microstates and $\auxsteps$ timesteps by the coarse-graining function $g$ and stochastic matrix $M$.
By definition, $M = \stablematrix^{(\auxsteps)}\stablematrix^{(\auxsteps-1)}\dots \stablematrix^{(1)}$ where each $\stablematrix^{(i)}$ is one-step. 
By \cref{lem:idem-equiv} and \cref{lem:prod-equiv}, there exists a matrix
\begin{align}
V = \tilde{\stablematrix}^{(\auxsteps)}\tilde{\stablematrix}^{(\auxsteps-1)}\dots \tilde{\stablematrix}^{(1)}
\label{eq:mprime}
\end{align} which obeys $\adj[V] \subseteq \adj[M]$, and where each $\tilde{\stablematrix}^{(i)}$ carries out an idempotent function.

Condition \ref{cond:sum} of \cref{def:cg} states that $\sum_{j \in g^{-1}(\fm(g(i))} M_{ji} = 1$ for all $i \in \dom(g)$, or equivalently that $\sum_{j \not\in g^{-1}(\fm(g(i))} M_{ji} = 0$.
Since the set of zero entries in $V$ is a superset of those in $M$, it is easy to see that if $M$ satisfies Condition 3, then so must $V$. Thus, $g$ and $V$ also implement $\hat{f}$ with $n$ microstates and $\auxsteps$ timesteps.

Let $\gamma : \Y \rightarrow \Y$ represent the idempotent function carried out by $\tilde{\stablematrix}^{(\auxsteps)}$.  Define the set
\begin{align}
D := \img(\gamma) \cap \dom(g) \,,
\label{eq:dset-def}
\end{align}
so that $D$ is the set of microstates which are in the image of $\gamma$ and which have a macrostate defined. Note that the image of any idempotent function consists only of fixed points of that function. Since $D \subseteq \img(\gamma)$, $D$ thus contains only fixed points of $\gamma$.

We now define a one-to-one function $\om : \Xm\rightarrow \Y$ from macrostates to microstates which maps every macrostate $z$ to one particular ``canonical" microstate contained in that macrostate. Formally, we require $\om$ to obey the following two conditions (any $\om$ which obeys these conditions suffices): 
\betight
\item For all $z\in g(D)$, $\om(z) \in g^{-1}(z) \cap D$ (i.e., every macrostate that has a microstate in $D$ is mapped by $\om$ to one of its own microstates in $D$)
\item For all $z\not\in g(D)$, $\om(z) \in g^{-1}(z)$ (i.e., every microstate that does not have a microstate in $D$ is mapped to one of its own microstates)
\eetight
Note that $\om$ is one-to-one since the sets $g^{-1}(z)$ are non-overlapping for different $z$.
Note also that for any $y\in \img(\om)$, $\om^{-1}(y) = g(y)$.

We now construct a ``modified" function $\gamma' : \Y \rightarrow \Y$ in the following manner,
\begin{align}
\gamma'(y) = \begin{cases}
\om(g(\gamma(y))) & \text{if $\gamma(y) \in D$} \\
\gamma(y) & \text{otherwise}
\end{cases}
\label{eq:gammaprimedef}
\end{align}
In words, $\gamma'$ is similar to $\gamma$, but its outputs are canonical microstates where possible. Below, we show two things: first that $\gamma'$ is idempotent, and second that if we replace $\gamma$ by $\gamma'$, we will still implement $\fm$.

To show that $\gamma'$ is idempotent, we demonstrate that $\img(\gamma')$ consists only of fixed points of $\gamma'$. To do so, we consider two cases separately:
\betight
\item $y\in \Y$ with $\gamma(y) \not\in D$, for which $\gamma'(y) = \gamma(y)$. Note that since $\gamma(y)$ is idempotent, $\gamma(\gamma(y)) = \gamma(y) \not\in D$, and therefore $\gamma'(\gamma'(y))=\gamma'(\gamma(y))=\gamma(\gamma(y))=\gamma(y)=\gamma'(y)$.

\item $y\in \Y$ with $\gamma(y) \in D$, for which $\gamma'(y) = \om(g(\gamma(y))$. In this case,
$g(\gamma(y)) \in g(D)$, so by construction $\om(g(\gamma(y))) \in D$, thus
$\gamma'(\gamma'(y)) = \gamma'(\om(g(\gamma(y)))) = \om(g(\gamma(\om(g(\gamma(y))))))$.
As mentioned above, all elements in $D$ are fixed points of $\gamma$, so we can write 
$\gamma(\om(g(\gamma(y)))) = \om(g(\gamma(y)))$ to give
$\gamma'(\gamma'(y)) = \om(g(\om(g(\gamma(y)))))$.
Furthermore, by construction of $\om$, $\om(z) \in g^{-1}(z)$, thus $g(\om(g(\cdot)) = g(\cdot)$, so we can further rewrite $\gamma'(\gamma'(y)) = \om(g(\gamma(y)) = \gamma'(y)$.
\eetight
This proves that $\gamma'$ is idempotent.

We now show that we still implement $\fm$ if instead of the last matrix carrying out $\gamma$, it instead carries out $\gamma'$. Let $\tilde{\stablematrix}'^{(\auxsteps)}$ be the one-step matrix that encodes function $\gamma'$, and define the stochastic matrix
\begin{align}
W = \tilde{\stablematrix}'^{(\auxsteps)}\tilde{\stablematrix}^{(\auxsteps-1)}\dots \tilde{\stablematrix}^{(1)}
\label{eq:wdef}
\end{align}
Now consider any $i \in \Y$, and let $j$ indicate the output state such that $V_{ji} = 1$, where $V$ is as defined in \cref{eq:mprime}.  Let $j'\in \Y$ be the final state such that $W_{j',i}=1$. We now note two things:
\betight
\item[(a)] By the definition of $\gamma'$ in \cref{eq:gammaprimedef}, it must be that either $j' = j$ (in case $j \not\in D$) or $j' = \om(g(j))$ (in case $j\in D$).  In either case, $g(j') = g(j)$ (in the former case trivially, and in the latter case since $g(\om(g(\cdot)) = g(\cdot)$, as mentioned before). It is easy to verify that if Condition~\ref{cond:sum} of \cref{def:cg} holds for $V$, it must also hold for $W$; thus, $W$ in \cref{eq:wdef} also implements $\fm$ with $\nm$ microstates and $\auxsteps$ timesteps.
\item[(b)] Consider the case when $i \in \dom(g)$ (i.e., the initial state belongs to some macrostate).  In that case, $j \in \dom(g)$ by Condition \ref{cond:sum} of \cref{def:cg}. In addition, $j$ is clearly always within $\img(\gamma)$.  Thus, when $i \in \dom(g)$, $j \in D$ (by \cref{eq:dset-def}) and $j' \in \img(\om)$ (by \cref{eq:gammaprimedef}).
\label{cond:2}
\eetight

Finally, define $\X := \img(\om)$ (i.e., the set of ``canonical" microstates). By definition of $\om$, $\X \subseteq \dom(g)$ (and therefore also $\X\subseteq \Y$). Note also that $|\Xm| = |\X|$, since $\om$ is one-to-one.
Consider the restriction of $W$ to $\X$, which we indicate by $W^\X$. Since $W$ is a product of 0/1 valued stochastic matrices, both $W$ and its restriction $W^\X$ are 0/1 valued.  Furthermore, for any input state $i \in \X$, $i \in \dom(g)$; therefore, by Condition~{(b)} in the above list, the $j$ that satisfies ${W}_{ji}=1$ itself obeys $j \in \X$. Combining these results with Condition~\ref{cond:sum} of \cref{def:cg} states that $W^\X$ is a valid stochastic matrix that carries out
\begin{equation}
W^\X_{ji} = \delta(j, \omega(\hat{f}(g(i)))) = \delta(j, \omega(\hat{f}(\omega^{-1}(i)))) \,,
\end{equation}
where we've used the fact that $\omega^{-1}=g$ over $\X$.
\end{proof}

\section{Restricted set of idempotents}

\label{restrictedidempotents}

To illustrate some of the issues a restriction on which idempotents can be implemented raises, consider the case where our full system is a set of $N$ visible spins plus an unspecified set of hidden spins. Suppose the only idempotent functions we can apply to our system are those that affect either one or two spins at a time, leaving all the others unchanged. Physically, this would mean that the Hamiltonian of our system is a sum of one-spin and two-spin terms. (We then implement an idempotent function by dynamically altering the relative strengths of those terms.)

We can implement any function over the set of $N$ spins using this set of idempotent
functions if the set of hidden spins is large enough --- so long as the idempotent functions allow us to change any set of
one or two spins. (The analysis if we can only change pairs of spins that
are neighbors on a lattice, as in an Ising spin, is more complicated.)
To see this, first note that we can use such an idempotent function to copy the state of a spin into a different ``target'' spin.
By repeating this function with different target spins, we can make any desired number of copies of the original spin. Next, note that another of our allowed idempotent functions
maps any spin-pair $(x_1, x_2) \rightarrow (0, \cs{NAND}(x_1, x_2))$, i.e., evaluates the $\cs{NAND}$ of the two spins
and stores the result in the second spin. So if we make a copy of both $x_1$ and $x_2$, and then run this
$\cs{NAND}$ idempotent function on that pair of copy-bits, we will have implemented a 
full $\cs{NAND}$ gate whose input bits were $x_1$ and $x_2$ and whose output bit is $\cs{NAND}(x_1, x_2)$. (We will also have zeroed
the copy-bit that doesn't equal $\cs{NAND}(x_1, x_2)$, but that doesn't matter.)

Now $\cs{NAND}$ is a universal logical gate, meaning that we can implement
any Boolean function $f : \{0,1\}^N \rightarrow \{0,1\}^N$ by appropriately connecting $\cs{NAND}$ gates~\cite{mano2008logic} into one another. 
(In general, such an implementation will require that some of the gates have fanout greater than 1 --- 
but we can implement an arbitrary fanout, by repeated using our bit-copy idempotent function.) So by using enough hidden states and an appropriate set
of two-spin idempotent functions, we can evaluate the (arbitrary)  function $f$ of the $N$ visible spins, storing the
resultant output in $N$ of the hidden spins. At that point we can copy the (hidden) output back to the (visible) input bits, thereby completing 
the process of running $f$ on those input bits.

In general, implementing $f$ with this construction will require more hidden states and more hidden timesteps than would implementing it using arbitrary idempotent functions. 
However, calculating the associated increase in the space and time costs can be quite challenging. 
The time cost in our construction is given by the depth of the circuit of $\cs{NAND}$ gates and the fanouts of those gates. 
On the other hand, the number of hidden states is 
determined by the number and type of gates in that circuit. The analysis of how these quantities
and their tradeoff depends on the function $f$
is closely related to ongoing research in circuit complexity theory~\cite{arora2009computational,savage1998models}. 

Moreover, there seems to be no reason to believe that using our set of allowed idempotent functions to make circuits of $\cs{NAND}$ gates is the most efficient way to use them. In general there will be a complicated tradeoff between re-using hidden
spins to implement multiple gates (thereby reducing the total number of hidden spins needed) and increasing the number of gates that can be operated in parallel (which reduces the total number of timesteps).

\bibliographystyle{unsrtnat}

\end{document}